\newif\ifDRAFT 
\DRAFTfalse

\documentclass[11pt]{article}
\usepackage[margin=1in]{geometry}
\usepackage{graphicx} 
\usepackage[framemethod=TikZ]{mdframed}
\usepackage{amsfonts, amsmath, amssymb, amsthm}
\usepackage{mathrsfs}  
\usepackage{algorithm}
\usepackage[noend]{algpseudocode}
\usepackage{hyperref}
\usepackage{wrapfig}
\usepackage{array}
\usepackage{cleveref}
\usepackage[inkscapeformat=pdf]{svg}
\usepackage{multirow}
\usepackage{mathtools}
\usepackage{enumitem}
\usepackage{thm-restate}
\usepackage{tikz}

\usepackage{tabulary} 
\usepackage{booktabs}
\usepackage{arydshln}

\ifDRAFT
    \usepackage{lineno}
    \linenumbers
\fi

\theoremstyle{plain}
\newtheorem{theorem}{Theorem}
\newtheorem{lemma}[theorem]{Lemma}

\newtheorem{observation}[theorem]{Observation}
\newtheorem{claim}[theorem]{Claim}

\theoremstyle{definition}
\newtheorem{definition}{Definition}

\newtheorem{problem}{Problem}

\crefname{equation}{Eqn.}{Eqns.}

\DeclareSymbolFont{bbold}{U}{bbold}{m}{n}
\DeclareSymbolFontAlphabet{\mathbbold}{bbold}

\newcommand{\IGNORE}[1]{}

\newcommand{\depth}{\operatorname{depth}}
\newcommand{\LCA}{\operatorname{LCA}}
\newcommand{\FLCA}{\operatorname{FLCA}}
\newcommand{\covers}{\succcurlyeq}

\newcommand{\lca}{\mathtt{lca}}
\newcommand{\anc}{\mathtt{anc}}
\newcommand{\undefined}{\mathtt{undefined}}

\title{Fault-Equivalent Lowest Common Ancestors}
\author{
    Asaf Petruschka%
    \thanks{Supported by an Azrieli Foundation fellowship, and by Merav Parter's grant from the European Research Council (ERC) under the European Union’s Horizon 2020 research and
    innovation programme, agreement No. 949083.}\\
    Weizmann Institute of Science\\
    \texttt{asaf.petruschka@weizmann.ac.il}
}

\date{}

\begin{document}

\maketitle

\begin{abstract}
    Let $T$ be a rooted tree in which a set $M$ of vertices are marked.
    The \emph{lowest common ancestor (LCA)} of $M$ is the unique vertex $\ell$ with the following property: after failing (i.e., deleting) any single vertex $x$ from $T$, the root remains connected to $\ell$ if and only if it remains connected to some marked vertex.
    In this note, we introduce a generalized notion called \emph{$f$-fault-equivalent LCAs ($f$-FLCA)}, obtained by adapting the above view to $f$ failures for arbitrary $f \geq 1$.
    We show that there is a unique vertex set $M^* = \FLCA(M,f)$ of minimal size such after the failure of any $f$ vertices (or less), the root remains connected to some $v \in M$ iff it remains connected to some $u \in M^*$.
    Computing $M^*$ takes linear time.
    A bound of $|M^*| \leq 2^{f-1}$ always holds, regardless of $|M|$, and holds with equality for some choice of $T$ and $M$.
\end{abstract}

\section{Introduction and Results}

Consider the following motivating problem.
There is an $n$-vertex tree $T$, rooted at a source $s$.
In this tree, a nonempty and possibly large vertex subset of interest $M \subseteq V(T)$ is \emph{marked}.
We are preparing for the future \emph{failures} (or \emph{faults}) of at most $f$ currently unknown vertices, which will be deleted from the tree.
(A faulty vertex may or may not be marked.)
After the failures $F \subseteq V(T)$ occur, we will be interested to understand whether they cause $s$ to disconnect from all the marked vertices.
Namely, we will want to answer the following question: Is there some marked vertex $v \in M$ that remains reachable from $s$ in $T - F$?

However, we would like to save on memory costs, and avoid storing the entire set $M$.
Instead, we want to preprocesses $M$ to find and store a smaller ``representative'' set of vertices $M^* \subseteq V(T)$, which is equivalent to $M$ in terms of the above question.
Namely, for every fault-set $F \subseteq V(T)$ with $|F| \leq f$, all vertices in $M^*$ are disconnected from $s$ in $T-F$ if and only if this is true for the original marked set $M$.

We now introduce some definitions to formalize the above.
First, we define the \emph{covering} relation.

\begin{definition}\label{def:cover}
    For two vertex sets $A, B \subseteq V(T)$, we say that $A$ \emph{covers} $B$, and denote $A \covers B$, if for every $b \in B$ there exists $a \in A$ which is an ancestor of $b$ in $T$.
\end{definition}

Note that the failure of $F \subseteq V(T)$ disconnects $s$ from all of $M$ iff $F \covers M$.
Our requirement from the representative set $M^*$ is that this should happen iff $F \covers M^*$, whenever $|F| \leq f$.
We therefore define the \emph{$f$-fault equivalence} relation.

\begin{definition}
    For $f \geq 1$, we say that two vertex sets $M,N \subseteq V(T)$ are \emph{$f$-fault-equivalent}, and denote $M \sim_f N$, if for any $F \subseteq V(T)$ with $|F| \leq f$, it holds that $F \covers M$ iff $F \covers N$.
\end{definition}

Our motivating problem can now be succinctly stated as follows:

\begin{problem}\label{problem}
    Given the tree $T$, the marked set $\emptyset \neq M \subseteq V(T)$ and a fault parameter $f \geq 1$, find a set $M^* \subseteq V(T)$ of minimal size such that $M^* \sim_f M $.
\end{problem}

\paragraph{Relation to Lowest Common Ancestors.}
When preparing for a single vertex failure, i.e., when $f = 1$, a moment's reflection will show that one can always choose $M^*$ having only one vertex: the \emph{lowest common ancestor (LCA)} of all marked vertices, denoted $\LCA(M)$.
Indeed, a single failed vertex $v$ disconnects all of $M$ from $s$ iff $\{v\} \covers M$, namely iff $v$ is a common ancestor of all marked vertices, which happens iff $v$ is an ancestor of $\LCA(M)$, i.e., iff $\{v\} \covers \{\LCA(M)\}$.
In fact, $\LCA(M)$ is the \emph{only} single vertex satisfying this property.
This means that $\LCA(M)$ could be equivalently defined by the unique optimal solution to \Cref{problem} with $f=1$.

Thus, letting $f$ increase beyond $1$ yields a generalized notion of LCA given by the optimal solution to \Cref{problem}.
Also, as we will show, the optimal solution $M^*$ is always unique, and consists of LCAs of subsets of $M$.
For these reasons, we call $M^*$ the \emph{$f$-fault LCA} of $M$, and denote it by $\FLCA(M, f)$.
(So, by the above discussion, $\FLCA(M,1) = \{\LCA(M)\}$.)

\paragraph{Results.}
In this note, we give a simple algorithm to compute $\FLCA(M,f)$, the optimal solution for \Cref{problem}, and answer a natural question of interest (given our "memory savings" motivation): how small can $\FLCA(M,f)$ be?
When $f = 1$, we saw it has size $1$, regardless of how large $M$ is. 
This extends to a bound of $2^{f-1}$ on the size of $\FLCA(M,f)$, which is worst-case optimal
(i.e., for some choice of $T,M$, $|\FLCA(M,f)| = 2^{f-1}$).
Thus, when the fault parameter $f$ is constant, we can represent any marked vertex set $M$ (in the sense of $f$-fault equivalence) using only a constant number of representative vertices.
Our results are summarized in the following theorem.
\begin{theorem}\label{thm:FLCA}
    Let $T$ be an $n$-vertex tree rooted at vertex $s$, $M \subseteq V(T)$ be a non-empty set of marked vertices, and $f \geq 1$.
    The following hold:
    \begin{enumerate}
		\item There is a unique set $\FLCA(M,f)$ having minimal size among all the $f$-fault-equivalent sets to $M$, namely among $\{N \subseteq V(T) \mid N \sim_f M\}$. 
		\item It holds that $|\FLCA(M,f)| \leq 2^{f-1}$, and this bound is tight.
        That is, for some choice of $T$ and $M$, this holds with equality.
		\item There is an $O(n)$ time algorithm to compute $\FLCA(M,f)$ given $T$, $M$ and $f$.
        Further, after $O(n)$ time for preprocessing $T$, one can compute $\FLCA(M,f)$ within $O(|M|)$ time.
    \end{enumerate}
\end{theorem}

\paragraph{Edge Faults.}
We remark that considering failures of \emph{edges} instead of vertices, or even allowing a mixture of failing vertices and edges, does not change our results regarding $\FLCA(M,f)$.
To state this explicitly:
$\FLCA(M,f)$ is the unique vertex set $M^* \subseteq V(T)$ having minimal size such that for every $F \subseteq V(T) \cup E(T)$ of size $|F| \leq f$, in $T-F$ it holds that the root $s$ is connected to some $v \in M$ iff it is connected to some $u \in M^*$.
This is due to the fact that, in terms of connectivity to the root, the failure of an edge in a tree has the same effect as the failure of its lower endpoint.

\paragraph{Aggregation.}
It is easy to prove that the function $\varphi (\cdot) = \FLCA(\cdot, f)$ admits the following nice \emph{aggregation} property: $\varphi(A \cup B) = \varphi (A \cup \varphi(B))$.
Such aggregation properties are often exploited for efficient computations.
As a ``toy example'', suppose the marked set $M$ is revealed to us over time in batches $M_1, M_2, M_3, \dots$. 
Then we can save on memory in the time between batches $t$ and $t+1$, only (at most $2^{f-1}$) vertices in $M^*_t = \FLCA(M_1 \cup \cdots \cup M_t, f)$.
When $M_{t+1}$ arrives, we can use the aggregation property and compute $M^*_{t+1}$ as $\FLCA(M^*_t \cup M_{t+1}, f)$.

\paragraph{Potential Applications.}
The notions and results presented in this note were developed during research on fault-tolerant graph data structures, but eventually did not make their way into the final solutions.
Still, the author believes they could be of potential use in the field of fault-tolerant graph structures and algorithms, and hopes such applications would be found in the future.

\section{Proof of \Cref{thm:FLCA}}

The proof is by analyzing the following algorithm for computing $\FLCA(M,f)$.
The notation $T_v$ stands for the subtree of $T$ rooted at vertex $v$.

\begin{algorithm}[H]
    \caption{Algorithm $\mathcal{A}$ for computing $\FLCA(M,f)$}\label{alg:FLCA}
    \textbf{Input:} Rooted tree $T$, non-empty vertex set $M \subseteq V(T)$, integer $f \geq 1$ \\
    \textbf{Output:} Vertex set $\mathcal{A}(T,M,f) \subseteq V(T)$
    \begin{algorithmic}[1]
    \State $\ell \gets \LCA(M)$
    \If{$\ell \in M$} \label{line:l-in-M}
        \Return $\{\ell\}$
    \EndIf
    \State $u_1 , \dots, u_d \gets$ the children of $\ell$ with $M \cap V(T_{u_i}) \neq \emptyset$  \Comment{Note:\ $d \geq 2$ as $\LCA(M) = \ell \notin M$} \label{line:children-of-l}
    \If{$d > f$} \label{line:d-vs-f}
        \Return $\{\ell\}$
    \EndIf
    \State $M_1, \dots, M_d \gets M \cap V(T_{u_1}), \dots, M \cap V(T_{u_k})$
    \State \Return $\bigcup_{i=1}^d \mathcal{A}(T, M_i, f-d+1)$ \Comment{Note:\ $1 \leq f-d+1 \leq f-1$} \label{line:recurse}
    \end{algorithmic}
\end{algorithm}

We divide the proof into several claims regarding algorithm $\mathcal{A}$.
All of them are proved by strong induction on $f$. We denote the output as $M^* = \mathcal{A}(T,M,f)$, and in case \Cref{line:recurse} is executed, we also denote $M^*_i = \mathcal{A}(T, M_i, f-d+1)$.

Throughout, we will use extensively the following easy-to-observe properties of the covering relation $\covers$ from \Cref{def:cover}, without explicitly stating them.
The notation $T[u,v]$ stands for the (unique) tree path between vertices $u$ and $v$.

\begin{observation}\label{obs:cover}
	The covering relation $\covers$ from \Cref{def:cover} has the following properties:
	\begin{enumerate}
		\item It is reflexive and transitive. (Namely, $\covers$ is a \emph{preorder}.)
		\item $A \covers B \iff B \subseteq \bigcup_{a \in A} V(T_a) \iff \forall b \in B, \, A \cap T[s,b] \neq \emptyset$.
		\item If $A \covers B$, then for every $B' \subseteq B$ it holds that $A \covers B'$ and $\{\LCA(A)\} \covers \{\LCA(B')\}$.
		\item If $A_i \covers B_i$ for all $i$, then $\bigcup_i A_i \covers \bigcup_i B_i$.
		\item Assume $A \covers B$ and $B \subseteq V(T_v)$. Let $A'$ be any subset of $A$ obtained by removing some vertices lying outside of $T_v \cup T[s,v]$. Then $A' \covers B$.
	\end{enumerate}
\end{observation}

We start with an auxiliary lemma, stating that the output $M^*$ must lie between $\ell = \LCA(M)$ and $M$ in terms of covering.
\begin{lemma}\label{cl:between-ell-and-M}
	$\{\ell\} \covers M^* \covers M$.
\end{lemma}
\begin{proof}
    If $M^* = \{\ell\}$ this is trivial. Otherwise, the algorithm must have executed \Cref{line:recurse}.
    By the induction hypothesis, $\{\LCA(M_i)\} \covers M^*_i \covers M_i$ for all $i=1, \dots, d$.
    As $\{\ell\} \covers \{\LCA(M_i)\}$, we deduce that
    \begin{align*}
        \{\ell\} &\covers \bigcup_{i=1}^d M^*_i = M^* \\
        &\covers \bigcup_{i=1}^d M_i = M
    \end{align*}
    as required.
\end{proof}

We now turn to prove the first item of~\Cref{thm:FLCA}, by the following \Cref{cl:fault-equivalent} and \Cref{cl:unique-minimal}.

\begin{claim}[$f$-Fault Equivalence]\label{cl:fault-equivalent}
	$M^* \sim_f M$.
\end{claim}
\begin{proof}
    Let $F \subseteq V(T)$ with $|F| \leq f$.
    We should prove that $F \covers M^* \iff F \covers M$.

    \begin{description}
        \item[$(\Longrightarrow)$]
        Follows immediately from \Cref{cl:between-ell-and-M}.

        \item[$(\Longleftarrow)$]
        If  $F \covers \{\ell\}$ then $F \covers M^*$ by \Cref{cl:between-ell-and-M} and we are done.
        Assume now that $F \not\covers \{\ell\}$, i.e.\ $F \cap T[s,\ell] = \emptyset$.
        As $F \covers M$, it follows that $\ell \notin M$.
        Hence, the condition of \Cref{line:l-in-M} is not satisfied, and \Cref{line:children-of-l} must have been executed.
        Each subtree $T_{u_i}$ must intersect $F$, since otherwise, the subset $M_i$ of $M$ could not have been covered by $F$ (because $F \cap T[s,\ell] = \emptyset$).
        Because $T_{u_1}, \dots, T_{u_d}$ are disjoint, we see that $d \leq |F| \leq f$.
        This means that the condition of \Cref{line:d-vs-f} is not satisfied, hence \Cref{line:recurse} must have been executed.
        
        Let $F_i = F \cap V(T_{u_i})$.
        Note that $F_i \covers M_i$, because $F \covers M_i$ and all vertices in $F - F_i$ lie outside of $T_{u_i} \cup T[s,u_i]$.
        Also, each of the $d-1$ disjoint subtrees $\{T_{u_j} \}_{j \neq i}$ contains one vertex from $F-F_i$, and thus $|F_i| \leq |F| - (d-1) \leq f-d+1$.
        Since $M_i \sim_{f-d+1} M^*_i$ holds by the induction hypothesis, we obtain that $F_i \covers M^*_i$. We conclude that
        \[
            F \covers \bigcup_{i=1}^d F_i \covers \bigcup_{i=1}^d M^*_i = M^*.
        \]
    \end{description}
\end{proof}

\begin{claim}[Minimality and Uniqueness]\label{cl:unique-minimal}
    If $N \subseteq V(T)$ and $N \sim_f M$, then $|N| \geq |M^*|$, and equality holds iff $N = M^*$.
\end{claim}
\begin{proof}
    Consider first the case where $M^* = \{\ell\}$. 
    Then, as $M$ is non-empty, $N$ also cannot be empty, i.e.\ $|N| \geq 1 = |M^*|$.
    If equality holds, then $N$ contains a single vertex $v$.
    Trivially, $\{v\} \covers N$.
    Since $N \sim_f M$ (and $f \geq 1$) we obtain $\{v\} \covers M$.
    Thus, $v$ is a common ancestor of all vertices in $M$, so it must be an ancestor of $\ell = \LCA(M)$.
    However, it cannot be a strict ancestor of $\ell$, as then we would have $\{\ell\} \covers M$ and $\{\ell\} \not\covers \{v\} = N$, contradicting the assumption that $N \sim_f M$.
    Thus $v = \ell$, so $N = M^*$.
    	
    It remains to consider the case where $M^* \neq \{\ell\}$.
    Then the algorithm must have executed \Cref{line:recurse}, and the conditions of \Cref{line:l-in-M} and \Cref{line:d-vs-f} were not satisfied.
    Hence, $U =\{u_1,\dots,u_d\} \covers M$ and $d \leq f$.
    Since $N \sim_f M$ we get that $U \covers N$.
    Therefore, letting $N_i = N \cap V(T_{u_i})$,
    we have 
    $
    N = \biguplus_{i=1}^d N_i
    $
    (where $\uplus$ denotes disjoint union).
    
    We now observe that, for any $F' \subseteq V(T)$ with $|F'| \leq f-d+1$, it holds that
    \[
	F' \covers M_i \iff F' \cup (U - \{u_i\}) \covers M \iff F' \cup (U - \{u_i\}) \covers N \iff F' \covers N_i 
    \]
    where the middle `$\Longleftrightarrow$' holds as $|F' \cup (U - \{u_i\})| \leq f$ and $N \sim_f M$.
    This means that $N_i \sim_{f-d+1} M_i$.
    The induction hypothesis thus yields that $|N_i| \geq |M^*_i|$, and equality holds iff $N_i = M^*_i$.
    We deduce that
	\begin{equation}\label{eq:N-lower-bound}
		|N| = \sum_{i=1}^d |N_i| \geq \sum_{i=1}^d |M^*_i|.
	\end{equation}
    Now, as $M_i \subseteq V(T_{u_i})$, it follows (by \Cref{cl:between-ell-and-M} for $M^*_i = \mathcal{A}(T,M_i,f-d+1)$) that $M^*_i \subseteq V(T_{u_i})$.
    Hence, the union returned in \Cref{line:recurse} is disjoint, i.e.\ $M^* = \biguplus_{i=1}^d M^*_i$.
    Thus, the right-hand-side of~\Cref{eq:N-lower-bound} is equal to $|M^*|$, so we have shown that $|N| \geq |M^*|$.
    Furthermore, in light of \Cref{eq:N-lower-bound}, $|N|=|M^*|$ can hold only if for all $i=1,\dots,d$ we have $|N_i| = |M^*_i|$, and thus also $N_i = M^*_i$.
    So in this case,
    \[
        N = \biguplus_{i=1}^d N_i = \biguplus_{i=1}^d M^*_i = M^*
    \]
    as required.	
\end{proof}

Next, we prove the second item of~\Cref{thm:FLCA}.

\begin{claim}[Size Bound]\label{cl:tight-size-bound}
    $|M^*| \leq 2^{f-1}$. Further, for some choice of $T$ and $M$, equality holds.
\end{claim}
\begin{proof}
    If $M^* = \{\ell\}$ then the inequality is trivial.
    Otherwise, \Cref{line:recurse} must have been executed, so $M^* = \bigcup_{i=1}^d M^*_i$. Using the induction hypothesis, we obtain
    \[
    |M^*| \leq \sum_{i=1}^d |M^*_i| \leq d \cdot 2^{(f-d+1)-1} = \frac{d}{2^{d-1}} \cdot 2^{f-1} \leq 2^{f-1}.
    \]
    For a case where equality holds, consider $T$ being a full binary of height at least $f-1$, with all of its leaves marked as $M$.
    Then it is easy to verify that $M^* = \mathcal{A}(T,M,f)$ is the set of all $2^{f-1}$ vertices with depth $f-1$.
\end{proof}

Finally, we prove the third and last item of~\Cref{thm:FLCA}.

\begin{claim}[Implementation]\label{cl:runtime}
    The tree $T$ can be preprocessed in $O(n)$ time so that queries $(M,f)$ can be answered with $\FLCA(M,f)$ within $O(|M|)$ time.
\end{claim}
\begin{proof}
    
    

Algorithm $\mathcal{A}$ can be implemented in $O(n)$ time by dynamic programming on the tree $T$.
As an improvement, we show after $O(n)$ time for preprocessing $T$, one can answer queries $(M,f)$ by computing $\FLCA(M,f)$ within $O(|M|)$ time.
To this end, We build two classical data structures for (pairwise) LCA and level ancestor queries:
\begin{description}
    \item[$\lca(u,v)$]: returns the lowest common ancestor of vertices $u$ and $v$
    \item[$\anc(u,l)$]: returns the ancestor $v$ of $u$ such that $\depth(v) = l$ (or $\undefined$ if $\depth(u) < l$)
\end{description}
This requires $O(n)$ time, and queries can be answered in $O(1)$ time~\cite{BenderF00,BenderF04}.
Additionally, we create a lookup table $D$ for vertices, where $D[v]$ stores a ``switch bit'' initialized to zero, and a pointer initialized to a null value.
This concludes the preprocessing, taking $O(n)$ time.
Given a query $(M,f)$ with $M = \{v_1, \dots, v_{|M|}\}$, we now explain the implementation details for executing $\mathcal{A}(T,M,f)$ in $O(|M|)$ time.

    Computing $\ell$ takes $O(|M|)$ time by initializing $\ell \gets \lca(v_1, v_2)$ and updating $\ell \gets \lca(\ell, v_i)$ for $i= 3, \dots, |M|$.
    Computing $u_1, \dots, u_d$ and $M_1, \dots M_d$ is most of the work.
    We aim to store $u_1, \dots, u_d$ in a linked list $L$, and to store each $M_i$ in a linked list pointed to from $D[u_i]$.
    We sequentially process each $v_i \in M$ as follows.
    First, we compute $u \gets \anc(v_i, \depth(\ell)+1)$.
    Next, we do a lookup to $D[u]$.
    If the switch bit is $0$, we (i) change it to $1$, (ii) change the pointer of $D[u]$ to the head of a new linked list of length $1$ that stores $v_i$, and (iii) add $u$ to the linked list $L$.
    If the switch bit is $1$, we just add $v$ to the linked list pointed by $D[u]$ (that was created when the switch was turned on).
    Overall, this takes $O(|M|)$ time.
    
    To invoke the recursive calls,
    We scan the linked list $L$ containing $u_1, \dots, u_d$ (where this is their order in $L$).
    When processing $u_i$, we ``clean up'' the lookup table entry $D[u_i]$ by copying the pointer to the linked list containing $M_i$, then reverting $D[u_i]$ to its original state (switch bit $0$ and null pointer).
    We can now execute the recursive calls $\mathcal{A}(T,M_i,f-d+1)$.
    This ensures that the lookup table $D$ returns to its cleaned state after the current query $(M,f)$ is answered, allowing us to treat any future query $(M',f')$ in the same manner.
\end{proof}

\section*{Acknowledgments}
I am grateful to Merav Parter for stimulating discussions and helpful comments.

\bibliography{references}
\bibliographystyle{alpha}

\end{document}